\documentclass[12pt]{article}

\usepackage{epsfig}
\usepackage{amssymb}
\usepackage{amsfonts}

\usepackage{color}

\def\be{\begin{equation}}
\def\ee{\end{equation}}
\def\ba{\begin{array}{c}}
\def\ea{\end{array}}

\def\ben{$$}
\def\een{$$}

\newcommand{\bea}{\begin{eqnarray}}
\newcommand{\eea}{\end{eqnarray}}

\newtheorem{thm}{Theorem}

\newtheorem{lemma}[thm]{Lemma}

\newenvironment{proof}{\noindent {\bf Proof}}{\hfill$\square$\vspace{3mm}\endtrivlist}

\begin{document}

\titlepage


 \begin{center}{\Large \bf

Complex tridiagonal quantum Hamiltonians and
matrix continued fractions

  }\end{center}


 \begin{center}

\vspace{8mm}

  {\bf Miloslav Znojil} $^{1,2}$

\end{center}

\vspace{8mm}

  $^{1}$
 {The Czech Academy of Sciences,
 Nuclear Physics Institute,
 Hlavn\'{\i} 130,
250 68 \v{R}e\v{z}, Czech Republic, {e-mail: znojil@ujf.cas.cz}}


 $^{2}$
 {Department of Physics, Faculty of
Science, University of Hradec Kr\'{a}lov\'{e}, Rokitansk\'{e}ho 62,
50003 Hradec Kr\'{a}lov\'{e},
 Czech Republic}


\section*{Abstract}

Quantum resonances described by
non-Hermitian tridiagonal-matrix Hamiltonians $H$
with complex energy eigenvalues 
are considered.
The
possibility is analyzed of the evaluation of
quantities $\sigma_n$
known as the
singular values
of
$H$. What is constructed are
self-adjoint block-tridiagonal
operators $\mathbb{H}$ (with eigenvalues $\sigma_n$)
and their resolvents (defined
in terms of a
matrix continued fraction,
MCF). In an illustrative
application of the formalism
to the discrete version of conventional $H=-d^2/dx^2+V(x)$
with complex local $V(x) \neq V^*(x)$,
the numerical MCF convergence is found quick and,
moreover, supported also by a
fixed-point-based formal proof.

\subsection*{Keywords}.

non-Hermitian complex-symmetric tridiagonal-matrix  Hamiltonians $H$;

singular values $\sigma_n \in \mathbb{R}$ of resonances in open quantum systems;

$\sigma_n$s as eigenvalues of a block-tridiagonal Hermitian partner $\mathbb{H}$ of $H$;

matrix continued fraction resolvent of $\mathbb{H}$;

analytic proof of convergence;

\newpage

\section{Introduction}

Physical quantum systems described by
tridiagonal Hamiltonians
 \be
 H=
  \left[ \begin {array}{ccccc}
     a_1&b_1&0
 &\ldots&0
   \\
   c_2&a_2&b_2&\ddots
 &\vdots
   \\
 0
 &\ddots&\ddots&\ddots&0
   \\
 \vdots&\ddots&c_{N-1}&a_{N-1}&b_{N-1}
    \\
  0&\ldots&0&c_{N}&a_{N}
    \\
 \end {array} \right]\,
 \label{ufinkit}
 \ee
with $N \leq \infty$ and complex $a_k$, $b_k$ and $c_{k+1}$
(cf., e.g., the recent collection \cite{Marcelo} of examples)
can be separated, roughly speaking,
into two subcategories
in which the
Hamiltonian is assumed
Hermitian (or
at least Hermitizable {\it alias\,} quasi-Hermitian \cite{Geyer}),
or not.

In the former case
the spectrum is real and
its interpretation is standard
(see \cite{Messiah} or any other textbook on quantum mechanics).
Even when our
Hamiltonian (\ref{ufinkit}) is non-Hermitian
but Hermitizable,
multiple comprehensive reviews
of the comparatively minor
necessary modification of the theory
are already
available
(cf, e.g., \cite{Carl,ali,book}).

In the other, Hermiticity violating
cases with complex spectra (which are of our present interest)
the evolution
generated by Hamiltonian (\ref{ufinkit})
is non-unitary.
We are forced to speak, typically,
about quantum mechanics of
resonances and/or about open
quantum systems
which are, by definition, exposed to the influence of
an environment
\cite{Nimrod}.

The latter form of quantum theory is, arguably
(i.e., still under vivid discussion \cite{Nimrodb,Nimrodc}), more realistic
since in it,
a key role is played by
unstable states. They are characterized by the
energy eigenvalues which are not real,
$E_n \in \mathbb{C}$.
Their experimental localization
requires, therefore, subtle and sophisticated
techniques.
Complications also arise in the theory
because the relevant complex eigenvalues have to be
deduced from
a suitable
non-Hermitian
effective form of Hamiltonian.

Whenever one reveals that the energy is not real,
an important partial characteristics of the state
can be provided by the real and non-negative auxiliary
quantity
 $
 \sigma_n $
called singular value~\cite{SV}.
In what follows,
we are going to pay attention to a few technical aspects of
the determination of these quantities.

In the literature, the study of singular values
is mainly being developed
in the framework of pure mathematics.
In this setting,
several intimate connections
between complex  $E_n \in \mathbb{C}$
and real~$\sigma_n \in \mathbb{R}$
have recently been revealed and described by
Pushnitski with \v{S}tampach \cite{PS1}.
We found these results inspiring.
Unfortunately,
they were
based on a number of formal assumptions
(including, typically, the assumption of the boundedness
of the operators)
which would rather severely restrict their
applicability to
the description of
quantum resonances.
Thus, we decided to weaken some of the assumptions.
We revealed that
for a class of operators (\ref{ufinkit})
exhibiting symmetry
with respect to
transposition (cf. \cite{Garcia,Garciab,gar-put}
for dedicated reviews),
a
fairly nontrivial formal
correspondence
between the constructions of
$E_n $ and $\sigma_n$
still survives.


\section{The plan of the paper}

The presentation of our results will be preceded by
a preparatory section
\ref{ch2} in which
a purely formal appeal of an arbitrary tridiagonal-matrix
Schr\"{o}dinger operator with
Hamiltonian of Eq.~(\ref{ufinkit})
(with $N<\infty$ or $N=\infty$)
will be shown to lie in the possibility
of its factorization
in terms of analytic continued fractions.
In subsections \ref{kuskusi} and \ref{ch3} we will
recall several
analytic and numerical applications of this
idea.
We will point out that in both of these cases
the purpose of the factorization is a efficient
identification of the energy levels with the
poles of a continued-fractional
resolvent {\it alias\,} Green's function.
In the third subsection \ref{ius}
it is finally explained that
a mathematical key to the success
lies in the efficiency of fixed-point-based
proof of its convergence.

The core of our message is then presented in section \ref{ch4}.
We will address there several technical complications
which emerge due to the fact that the spectrum of
our effective Hamiltonian $H$ of Eq.~(\ref{ufinkit}) is complex.
In our paper we propose
two sources of simplification. Besides the above-mentioned
reduction of interest from the (complex) energies
$E_n$ to the mere (real) singular values $\sigma_n$,
we will also make use of an idea of
Pushnitski and \v{S}tampach~\cite{PS1}
and we will treat these singular values
as eigenvalues of
a certain auxiliary
block-tridiagonal
descendant $\mathbb{H}$ of $H$.
In subsection \ref{houska}
this will enable us to
replace the above-outlined universal
procedure of a rather difficult
analytic-continued-fraction
factorization
of the manifestly non-Hermitian resolvent $R(z)=(H-z)^{-1}$
with the complex poles at $z=E_n$
by an innovative simplified
procedure of an analogous
matrix-continued-fraction (MCF)
factorization
of the associated resolvent
$\mathbb{R}(z)=(\mathbb{H}-z)^{-1}$
with the real and positive poles at
the physical singular values of $H$.

It will be emphasized that
a key merit of
transition from $H$ to
$\mathbb{H}$ is that
the latter operator
is strictly Hermitian.
For this reason,
all of the
phenomenologically relevant
poles of the associated
resolvent
$\mathbb{R}(z)$
have to be sought on the real half-axis.
Naturally, the price to pay is the mere
block-tridiagonality of
$\mathbb{H}$ opening the
less elementary
question of
the MCF convergence. This problem will be
discussed and resolved in subsection \ref{umius}.
Finally, for illustration, several sets of parameters will be
chosen to show, in subsection \ref{ch5},
the efficiency of the MCF method
as well as the existence of the limitations
of its applicability.

A few comments and conclusions will be finally added
in section \ref{ch6}.

\section{Factorization of resolvents\label{ch2}}

In an overall framework of quantum mechanics one of the
key formal
advantages of the tridiagonality of Hamiltonians (\ref{ufinkit})
can be seen in the possibility of the following form of
factorization
of Schr\"{o}dinger operators,
 \be
 H-E
= {\cal U}\,
 {\cal F}\,
 {\cal L}\,\,
 \label{finkit}
 \ee
where the middle factor ${\cal F}$ is a diagonal matrix
with elements
 $$
 1/f_1,  1/f_2, \ldots,  1/f_N\,.
 $$
These elements have to be defined by recurrences
 \be
 f_k=\frac{1}{a_k-E-b_kf_{k+1}c_{k+1}}\,,\ \ \ \
 k=N, N-1,\ldots,2 ,1 
 \label{cf}
 \ee
in which we set, formally, $f_{N+1}=0$.
The factorization then becomes an identity when we
set
 \be
 {\cal U}=
  \left[ \begin {array}{ccccc}
  1&b_1f_2&0
 &\ldots&0
   \\
     0&1&b_2f_3&\ddots
 &\vdots
   \\
 0
 &0&\ddots&\ddots&0
   \\
 \vdots&\ddots&\ddots&1&b_{N-1}f_N
    \\
  0&\ldots&0&0&1
    \\
 \end {array} \right]\,,
 \ \ \ \ \ \
 {\cal L}=
  \left[ \begin {array}{ccccc}
  1&0&0
 &\ldots&0
   \\
     f_2c_2&1&0
 &\ldots&0
   \\
   0&f_3c_3&\ddots&\ddots
 &\vdots
   \\
 \vdots&\ddots
 &\ddots&1&0
   \\
  0&\ldots&0&f_Nc_{N}&1
    \\
 \end {array} \right]\,.
 \label{lowkit}
 \ee
In the case of operators with $N = \infty$
the factorization~(\ref{finkit}) becomes formally defined
in terms of
recurrences (\ref{cf}) and
quantities $f_k=f_k(E)$ called analytic continued fractions
\cite{Akhiezer}.
Naturally, these continued fractions must be convergent;
otherwise, the factorization (\ref{finkit}) would not exist of course.

\subsection{Characteristic application: Green's function\label{kuskusi}}

The idea of factorization
can immediately be extended
to the evaluation of resolvents \cite{Haydock},
 \be
 \frac{1}
 {H-E}=
 {\cal L}^{-1}\,
 {\cal F}^{-1}\,
 {\cal U}^{-1}\,.
 \label{[9h]}
 \ee
The construction is almost equally explicit since
 \ben
 {\cal U}^{-1}=
  \left[ \begin {array}{ccccc}
  1&{u}_2&{u}_2{u}_3
 &\ldots&{u}_2{u}_3\ldots{u}_N
   \\
     0&1&{u}_3&\ddots
 &\vdots
   \\
 0
 &0&\ddots&\ddots&{u}_{N-1}{u}_N
   \\
 \vdots&\ddots&\ddots&1&{u}_N
    \\
  0&\ldots&0&0&1
    \\
 \end {array} \right]\,,
 \ \ \
 {\cal L}^{-1}=
  \left[ \begin {array}{ccccc}
  1&0&0
 &\ldots&0
   \\
     {v}_2&1&0
 &\ldots&0
   \\
   {v}_3{v}_2&{v}_3&\ddots&\ddots
 &\vdots
   \\
 \vdots&\ddots
 &\ddots&1&0
   \\
   {v}_N \ldots {v}_3{v}_2&\ldots&{v}_N{v}_{N-1}&{v}_N&1
    \\
 \end {array} \right]\,
 \een
where we only abbreviated
${u}_{k+1}=-b_kf_{k+1}$ and ${v}_j=-c_jf_j$.

In many
physics-oriented applications
of formula (\ref{[9h]})
with $N=\infty$
one often needs to evaluate just the
so called analytic Green's function
$f_1(z)=[1/(H-z)]_{11}$ where $z$ (not in the spectrum of $H$)
is just a suitable complex
parameter.
For illustration, {\it pars pro toto},
we could cite paper \cite{Singh}
in which,
in a broader methodical context of quantum field theory,
the (real) poles of function $f_1=f_1(E)$
were shown to coincide with the
energy spectrum of the quantum-mechanical bound
states in the sextic-anharmonic-oscillator
local potential $V(x)=Ax^2+Bx^4+Cx^6$.

In {\it loc. cit.},
serendipitously,
Singh et al were probably also the first
physicists who discovered, as a byproduct
of their analysis, the so called
quasi-exact solvability property of their Hamiltonian-representing
toy-model (\ref{finkit}):
A more extensive account of the
details of this phenomenon
(related to the continued-fraction termination
due to an accidental disappearance of element $c_{k_0+1}=0$
at some anomalous coupling constants and index $k_0$)
can be
found in dedicated monograph \cite{Ushveridze}.
Still,
in the generic, non-terminating cases with $N=\infty$,
a key to the
consistency of all of the similar
results
appeared to lie in the proof of
convergence of the underlying analytic
continued fractions (cf., also \cite{FP}).

\subsection{Discrete complex local confining-potential example\label{ch3}}

In the context of purely numerical calculations in quantum physics,
the tridiagonal format of Hamiltonians (\ref{ufinkit})
often finds its origin
in the most common ordinary differential operator
 \be
 H= -\frac{d^2}{dr^2}+V(r)\,,\ \ \ \ r \in (0,\infty)
 \label{kukish}
 \ee
acting, say, in the most common Hilbert space $L^{2}(0,\infty)$.
Indeed, after
one replaces the half-line of
the continuous coordinates $r \in (0,\infty)$
by a discrete (though not necessarily equidistant) lattice
of grid points $r_k$ with
$k=0,1,\ldots,N,N+1$ \cite{grid},
the
related discretization of the
original kinetic-energy term $ -d^2/dr^2$ becomes
responsible for the emergence of real off-diagonal
matrix elements in (\ref{ufinkit}).
In the equidistant-lattice case one can even
choose $b_k=c_{k+1}=\alpha_k=1$ at all $k$.

Once we insist on keeping
the potential strictly local (albeit complex),
we only have a freedom in the choice of
the diagonal matrix elements in a grid-point-dependent
(i.e., in a subscript-dependent) form of
$V(r_k)=a_k = \beta_k+i\gamma_k$.
Our attention becomes restricted to the
complex and symmetric
$N=\infty$ matrix Hamiltonians
 \be
 H=
  \left[ \begin {array}{cccc}
     \beta_1+i\gamma_1&\alpha_1&0
 &\ldots
   \\
   \alpha_1&\beta_2+i\gamma_2&\alpha_2&\ddots
   \\
 0
 &\alpha_2&\beta_3+i\gamma_3&\ddots
   \\
 \vdots&\ddots&\ddots&\ddots
    \\
 \end {array} \right]\,.
 \label{vfinkit}
 \ee
Besides our physics-motivated reference to Eq.~(\ref{kukish})
it might make sense to
search also for an additional support of such an ansatz
in mathematically oriented
reviews \cite{Garcia,Garciab,gar-put}.

\subsection{Hermitian limit $\gamma_k \to  0$ and the
criteria of convergence of analytic continued fractions \label{ius}}

In many tridiagonal-matrix models
as sampled, typically, in  \cite{Singh,FP}
the size of
the diagonal matrix elements
$a_k$
happens to grow rather quickly with $k$.
As a consequence,
the proof of convergence of the continued-fraction
expansions of functions $f_k(z)$
gets simplified
in a way
based on an asymptotic rescaling
of the $k-$dependence of the dynamics-specifying matrix
elements $a_k$  and products $\rho_k=b_kc_{k+1}$.
This reduces the analysis of convergence to the study of iterations
of a less $k-$dependent version
of mapping (\ref{cf}). Thus, we may drop the subscripts and write
 \be
 f'=1/(\beta-E-\alpha^2f)\,.
 \label{fpcf}
 \ee
In addition,
when
one recalls
toy model (\ref{kukish}) one finds out that
also the $k-$dependence
of the off-diagonal matrix elements may be assumed smooth so that
it makes sense to
scale it out and fix, say, $\alpha = 1/\sqrt{2}$.

In the same methodical and
heuristic spirit
the quantities $\beta_k$
can be perceived as
real parts of a given discrete, smooth and confining, i.e., asymptotically growing
local potential $V(r_k)$.
At all of the  sufficiently large $k \gg 1$
their growth could be read as relation
$\beta_k\gg |E|$
so that the constant $E$
can
be neglected. This
simplifies the iteration recipe,
 \be
 f'=\frac{2}{2\beta-f}\,.
 \label{ipro}
 \ee
The proof of its convergence
is trivial but instructive,
proceeding in two steps.
In the first one we must find all of the eligible
fixed points  $f=f_{(FP)}$
of the mapping,
i.e., all of the roots of
equation $f_{(FP)}'=f_{(FP)}$.
This yields strictly
two candidates for the fixed point,
 \be
 f_{FP}^{(\pm)}=\beta\pm \sqrt{\beta^2-2}\,.
 \label{tyfps}
 \ee
They must be real because $\beta$ is real
so that the continued fractions
would certainly not converge when $\beta \in (-\sqrt{2},\sqrt{2})$.

In the second step of the proof we must take into account that
the iterations can only
converge
when the iterations
diminish the change, i.e.,
under the stability {\it alias\,} accumulation condition
 \be
 \left |\frac{\partial f'}{\partial f}\right | < 1\,.
\label{krikon}
 \ee
Fortunately, such a criterion is easy to apply since
 \be
 \left .\frac{\partial f'}{\partial f}\right |_{f=f_{FP}^{}}
 =\frac{1}{2}\,
 \left (f_{FP}^{}\right )^2 \,.
 \ee
Thus, the criterion of convergence (\ref{krikon})
is satisfied
near and only near the smaller fixed-point root.
Indeed,
we may set $\beta=\sqrt{2}(1+\delta^2)$ and
get $[f_{FP}^{(\pm)}]^2/2=1 \pm \sqrt{2}|\delta| +{\cal O}(\delta^2)$.
Naturally, the rate of convergence of the iterations
(and, hence, also the rate of convergence of the
related continued fractions)
will grow with the growth of $\delta$ or $\beta$.
Still, even in the domain of
very large $\beta$ the unique stable limiting
value $f_{FP}^{(-)}=1/\beta+ corrections$
will remain
positive.

\section{Singular values of Hamiltonian \label{ch4}}

Let us now return to the full-fledged
complex model~(\ref{vfinkit}) with
$\gamma_k \neq 0$. As long as
its main diagonal is complex,
we have to expect that in general the spectrum becomes complex
as well.
From the point of view of physics the eigenvalues
$E_n \in \mathbb{C}$ may only
represent unstable resonant states.

From a mathematical perspective
the factorization of the resolvent becomes less useful.
In particular,
the elementary version of the fixed-point-based method of
proof of the continued-fraction convergence
as outlined in subsection \ref{ius}
will cease to be applicable.
For all of these reasons it makes sense to replace
the over-ambitious task of the localization of
the energies by the mere determination
of the real and non-negative
singular values.

One of benefits of such a change of eigenvalue problem
can be seen in a tentative replacement
of the factorization of our tridiagonal
Schr\"{o}dinger operator $H-E$
by an analogous factorization of
its block-tridiagonal partner ${\mathbb H}-E$.

\begin{lemma}
At any
finite or infinite Hilbert-space dimension $N \leq \infty$
the singular values
of our complex symmetric
Hamiltonian matrix (\ref{vfinkit})
can be calculated as eigenvalues  of
an auxiliary block-tridiagonal Hermitian matrix
 \be
 {\mathbb H}=
  \left[ \begin {array}{ccccc}
     A_1&B_1&0
 &\ldots&0
   \\
   C_2&A_2&B_2&\ddots&\vdots
   \\
 0
 &C_3&A_3&\ddots&0
   \\
 \vdots&\ddots&\ddots&\ddots&B_{N-1}
    \\
 0&\ldots &0
 &C_N&A_N
   \\
 \end {array} \right]\,
 \label{matrix}
 \ee
where
 \be
 A_k=
\left (
\begin{array}{cc}
0&\beta_k+i\gamma_k\\
\beta_k-i\gamma_k&0
\ea
\right )\,,\ \ \ \ \
B_k=C_{k+1}=
\left (
\begin{array}{cc}
0&\alpha_k\\
\alpha_k&0
\ea
\right )\,
 \label{[14]}
 \ee
at all $k$.
\end{lemma}
\begin{proof}
. The real
(and, necessarily, non-negative) singular
values $\sigma_n$ of $H$
can be, for our present purposes,
defined as
eigenvalues
of the following auxiliary
Hermitian descendant
 \be
  \widetilde{\mathbb H}=
  \left (\begin{array}{cc}
  0&H\\H^\dagger&0
  \end{array}
  \right )\,
  \label{desce}
  \ee
of our non-Hermitian
Hamiltonian.
The spectral equivalence between $\widetilde{\mathbb H}$
and ${\mathbb H}$ is then an immediate consequence
of the
Pushnitski's and \v{S}tampach's
renumbering of the elements of the basis (see Ref.~\cite{PS1}).
\end{proof}

The
triplets of
the separate quantum-dynamics-representing parameters
$\alpha_k$, $\beta_k$ and $\gamma_k$
will be assumed real at all $k$.
Due to the loss of the
Hermiticity of $H$ the spectrum must be
also expected complex in general.
Its immediate phenomenological
applicability
may be sought
in the theory
of open quantum systems and resonances
in which one may often decide to replace the study of
the energies
by the study of the mere singular values.

In
functional analysis
the rigorous definition
of singular values
usually
requires
a compactness of $H$
\cite{SV}.
In our present
study of a class of quantum models
with tridiagonal complex Hamiltonians  of
the rather general form of Eq.~(\ref{ufinkit})
we will proceed in a more pragmatic manner and we will
assume that the
spectrum of $H$ is
discrete
and non-degenerate,
representing, say, an experimentally localizable
set of resonances. In such a setting, the assumption
of boundedness of $H$ would be counterproductive.

\subsection{Matrix continued fractions\label{houska}}

After the replacements $f_k \to F_k$, $a_k \to A_k$, $b_k \to B_k$ and
$c_{k+1} \to C_{k+1}$ of the real numbers by Hermitian two-by-two
matrices in (\ref{ufinkit}),
the
structure of the separate partitioned factors
${\cal U}$,
 ${\cal F}$ and
 ${\cal L}$
remains, after the block-tridiagonal partitioning
of ${\mathbb H}-\sigma$, the same.
The modified
middle factor ${\cal F}$ becomes a block-diagonal matrix
with elements $F_k^{-1}$ which have to be
defined by the
matrix-continued-fraction recurrences
 \be
 F_k=\frac{1}{A_k-\sigma-B_k F_{k+1}C_{k+1}}\,,\ \ \ \
 k=N, N-1,\ldots,2 ,1\,.
 \label{macf}
 \ee
The initial two-by-two matrix $F_{N+1}$ is to be set equal to
zero matrix
(interested readers can check, e.g., papers \cite{GG,GM,mocf,Nex}
for a more detailed reference).

Parameters $\alpha_k$, $\beta_k$ and $\gamma_k$
entering the
two ``input information'' matrices (\ref{[14]})
are real so that at every index $k$
we may
introduce other four real numbers and
reparametrize our MCF matrices,
 \be
 F_k^{-1}=
\left (
\begin{array}{cc}
u_k&x_k+iy_k\\
x_k-i y_k&v_k
\ea
\right )\,,\ \ \ \
 F_k=\frac{1}{u_kv_k-x_k^2-y_k^2}
\left (
\begin{array}{cc}
v_k&-x_k-i y_k\\
-x_k+iy_k&u_k
\ea
\right )\,.
 \label{rMcf}
 \ee
In this notation the two-by-two complex-matrix mapping $F\ (=
F_{k+1})\ \to \ F'\ (=F_k)$ can be reinterpreted as a quadruplet of
mutually coupled scalar maps
 \be
 u'=-\sigma-\alpha^2\frac{u}{uv-x^2-y^2}\,,
 \ \ \ \
 v'=-\sigma-\alpha^2\frac{v}{uv-x^2-y^2}\,,
 \label{[19]}
 \ee
 \be
 x'=\beta+\alpha^2\frac{x}{uv-x^2-y^2}\,,
 \ \ \ \
 y'=\gamma-\alpha^2\frac{y}{uv-x^2-y^2}\,.
 \ee
where we dropped the subscripted index $k$ as inessential.

\subsection{Convergence\label{umius}}

We may
notice that
$v=u$ so that the second mapping in (\ref{[19]})
is redundant. We are
left with the slightly simplified triplet
of iteration recipes
 \be
 u'=-\sigma-\alpha^2\frac{u}{u^2-x^2-y^2}\,,
 \ \ \ \
 x'=\beta+\alpha^2\frac{x}{u^2-x^2-y^2}\,,
 \ \ \ \
 y'=\gamma-\alpha^2\frac{y}{u^2-x^2-y^2}\,
 \label{iterr}
 \ee
with $u_{initial}=-\sigma \leq 0$,
$x_{initial}=\beta$ and $y_{initial}=\gamma$.
Moreover, we see that
 \be
 \gamma\,u=-\sigma\,y
 \label{lire}
 \ee
so that whenever $\sigma \neq 0 \neq \gamma$,
the iterative evaluations of the
sequence of the imaginary MCF components $\ iy$
can be also dropped as superfluous.

%

In subsection \ref{ius} above we demonstrated that
the analysis of convergence of the
analytic continued fraction expansions
can be reduced to the analysis of convergence of
iterations of the map $f \to f'$ sampled
by Eq.~(\ref{ipro}).
We employed there a straightforward
geometric interpretation of the mapping.
Then it was easy to find all of the
fixed points
(i.e., all of the possible accumulation points, cf. Eq.~(\ref{tyfps})).
In addition, it was also easy to
list all of the stable ones.

Now, we intend to apply the same method
to the study of convergence of MCFs (\ref{rMcf}).
The same constructive
philosophy is to be used.
Thus,
we may fix the scale (by choosing, say,
$\alpha=1$) and we have to list,
also in the case of
two-by-two MCF
expansions, all of the real
fixed points such that $u'=u$, $x'=x$ and $y'=y$
in Eqs.~(\ref{iterr}).

%

\begin{lemma}
For Hamiltonian (\ref{vfinkit})
the asymptotic $k \gg 1$ fixed-point approximation
 \be
 F_k^{-1}=
\left (
\begin{array}{cc}
u_k&x_k+iy_k\\
x_k-i y_k&v_k
\ea
\right )
=
\left (
\begin{array}{cc}
u&x+iy\\
x-i y&v
\ea
\right ) + {\rm corrections}\,,
\ \ \ \ k \gg 1
 \label{FPrMcf}
 \ee
of our auxiliary MCF matrices, if it
exists, can be defined, at a rescaled $\alpha=1$
and with $v=u$, by a real root $u=u_{(FP)}$
of quartic polynomial
 \ben
 P(u)=-4\,{u}^{4}{{\sigma}}^{2}
 +4\,{u}^{4}{{\gamma}}^{2}-8\,{{\sigma}}^{3}{u}^{3}
 +8\,{\sigma}{u}^{3}{{\gamma}}^{2}
 +{{\sigma}}^{2}{u}^{2}{{\beta}}^{2}-5\,{{\sigma}}^{4}{u}^{2}+
   \een
   \be
   +5\,{u}^{2}{{\gamma}}^{2}{{\sigma}}^{2}-4
 \,{{\sigma}}^{2}{u}^{2}+{{\sigma}}^{3}u{{\beta}}^{2}
 -4\,u{{\sigma}}^{3}-{{\sigma}}^{5}u+u{{\gamma}}^{2}{{\sigma}}^{3}-
 {{\sigma}}^{4} 
 \,.
 \ee
The value of
$y=y_{(FP)}(u_{(FP)})$ is
specified by Eq.~(\ref{lire})
and the value of $x=x_{(FP)}(u_{(FP)})$ is
given by another closed-form relation
 \be
 {\beta}{{\sigma}}^{3}x-4\,{u}^{3}{{\sigma}}^{2}
 +4\,{u}^{3}{{\gamma}}^{2}-6\,{{\sigma}}^{3}{u}^{2}+6\,{u}
^{2}{{\gamma}}^{2}{\sigma}-4\,u{{\sigma}}^{2}
 +{{\sigma}}^{2}u{{\beta}}^{2}-2\,u{{\sigma}}^{4}
 +2\,u{{\sigma}}^{2}{{\gamma}}^{2}-2\,{{\sigma}}^{3}=0
 \label{[24]}
 \ee
which is linear in $x$ and in which we abbreviated $u=u_{(FP)}$ .
\end{lemma}
\begin{proof}
is based on the standard
elimination of $x=x_{(FP)}$ using the
concept of the so called Gr\"{o}bner basis.
\end{proof}

\subsection{Illustrative examples\label{ch5}}

A truly remarkable feature of
the complex symmetric matrix model (\ref{vfinkit})
is that all of the fixed points of the related MCF
mapping $F \to F'$ can be defined via quartic
algebraic equation $P(u_{(FP)})=0$,
i.e., in closed form,
in principle at least.
One should be surprised by
the elementary non-numerical nature of such a result.
At the same time,
it makes probably no sense to offer here
a routine description and/or an exhaustive discussion
of the explicit multiparametric criteria of convergence.
We believe that
for our present
purposes it will be sufficient to
pick up just a few illustrative examples
using a (more or less arbitrary) set of
preselected parameters.

\begin{table}[th]
\caption{Sample of the quick numerical MCF
convergence. Mappings (\ref{iterr}) are iterated using constant,
$k-$independent
parameters $\alpha=\sigma=1$, $\beta=4$ and $\gamma=1/2$.
 } \label{pexp3b}
\begin{center}
\begin{tabular}{||c|ccc||}
\hline \hline
  iteration &
  \multicolumn{3}{c||}{\rm MCF element }\\
  & u& x& y\\
 \hline \hline
                    0      &   -1.000000000
                           &  4.000000000
  &                           0.5000000000\\
                     1     &   -1.065573770
                           &  3.737704918
  &                           0.5327868852\\
  2 &                          -1.081224617
    &                         3.715089035
     &                        0.5406123086\\
   3   &                       -1.083653085
       &                      3.712567903
        &                     0.5418265425\\
    4     &                    -1.083988278
          &                   3.712258295
           &                  0.5419941392\\
     5       &                 -1.084032778
             &                3.712218864
              &               0.5420163892\\
      6         &              -1.084038607
                &             3.712213775
                 &            0.5420193033\\
       7           &           -1.084039366
                  &          3.712213115
                    &         0.5420196832\\
        8           &        -1.084039465
                      &       3.712213029
                      &      0.5420197326\\
         9               &     -1.084039478
                         &    3.712213018
                          &   0.5420197391\\
          10                 &  -1.084039480
                            & 3.712213017
                             &0.5420197399\\
  \hline \hline
\end{tabular}
\end{center}
\end{table}
%
%
%
%
%
%
%
%
%
%
%
%
%

Let such an {\it ad hoc\,} choice be
 \be
 {\sigma}=1\,,\ \ \ {\gamma}=1/2\,,\ \ \ {\beta}=4\,.
 \label{spedi}
 \ee
Then, equation $P(u)=0$ yields the following four
exact fixed-point roots
 $$
 u_{(\pm,\pm)}=-\frac{1}{2}\, \pm \frac{1}{12}
 \,\sqrt {306 \pm 6\,\sqrt {1833}}$$
%
%
i.e., numerically, a quadruplet of real numbers
 $$
 \{-2.477093292, -1.084039480, 0.08403948007, 1.477093292\}\,.
 $$
The stable point of accumulation
can be shown to be the second one in the list.
This is confirmed by Table \ref{pexp3b} where
we can see that
for our choice (\ref{spedi}) of specific dynamical-input parameters
even the practical numerical
rate of convergence of our MCF-simulating iterations
of mappings (\ref{iterr}) is fairly quick.

The latter choice has intuitively been supported by the
expectation that the well-behaved Hamiltonian $H$ should
be dominated by its main diagonal.
Unfortunately, such a form of intuition is by far not
reliable.
For a proof it is sufficient to make just a
minor modification of $H$ and choose
 \be
 {\sigma}=1\,,\ \ \ {\gamma}=1/2\,,\ \ \ {\beta}=2\,.
 \label{expedi}
 \ee
%
In such a case we will have to deal with another quartic
fixed-point-determining polynomial
$P(u)=4+12\,{u}^{4}+24\,{u}^{3}+15\,{u}^{2}+3\,u$
which is strictly positive.
All of its roots
are of the closed form again,
 \be
 u^{(FP)}_{(\pm,\pm)}=-\frac{1}{2} \pm \frac{1}{12}
 \,\sqrt {18\pm 6\,i\sqrt {183}}
 \ee
but they are all complex,
 \be
 \{-1.092600316 \pm 0.4755787365\,i, 0.09260031606 \pm 0.4755787365\,i
 \}\,.
 \ee
Thus,
the iterations of mappings (\ref{iterr})
do not converge and may be shown to
yield only irregular and oscillatory
(i.e., especially for large $N \gg 1$, absolutely useless)
numerical results.

Obviously, even the
methodically motivated
comparison of the utterly different consequences of
assumptions (\ref{spedi})
and (\ref{expedi})
indicates that
the researchers which would wish to implement the method
would have to expect the emergence of several challenges in practice.
Among them,
the possible slow-down of the
rate of convergence (i.e., an increase
of the computational time)
would still remain to be just a
minor one. Indeed, due to
the rather elementary matrix
form of our present
Hamiltonians (\ref{ufinkit}), all of
the related numerical tests would still remain quick and conclusive
in virtually all
(viz., convergent, slowly convergent or divergent)
dynamical regimes.
The reason is that the number of parameters
in the underlying continued-fraction mappings
would not still be too large.
In this sense,
the true technical obstacles
could only be expected to emerge after
a generalization of the approach from the tridiagonal
models (\ref{ufinkit})
to their various
(i.e., typically, complex block-tridiagonal and manifestly asymmetric) generalizations
(cf., e.g., \cite{Nimrodd}).
Naturally, the study of these generalizations would
already lie far beyond the scope of our present paper.

\section{Discussion\label{ch6}}

The appeal of tridiagonality
ranges from its
numerical merits
to various analytic aspects and consequences.
In the former context
people often pre-tridiagonalize  general
matrices before they start searching, say, for their
eigenvalues. In the opposite, strictly non-numerical
and analytic extreme
the tridiagonality of certain matrices
plays a key role in the
abstract theory of orthogonal polynomials \cite{OgPol}.

Such a split of roles
survives when
one moves to the matrices and operators with complex spectra.
Still, the solution of complex discrete
Schr\"{o}dinger equations
need not be an easy or routine numerical task.
Similarly, challenges are also encountered
in a strictly analytic framework where, for example,
the authors of paper \cite{PS2}
discovered that
a transition to the complex symmetric
matrices of coefficients can lead to
an innovative
notion
of ``anti-orthogonal polynomials''
(their explicit illustrative Chebyshev-like
sample can be found in section Nr. 5.6 of {\it loc. cit.}).

Similar papers are helping to bridge the gap between numerical and
analytic points of view.
Opening the way towards
generalizations in which, for example,
the real or complex matrix elements become
replaced by the  real or complex $M$ by $M$ submatrices.
In this framework it would be also possible to
cover the systems
in which the tridiagonal matrix structure is modified:
{\it Pars pro toto\,} let us mention paper \cite{Nimrodd}
in which the
authors considered the direct sums of tridiagonal-matrix
operators
with applicability connected, e.g., with
the study of the so called squeezed or bi-squeezed
states.

In such a setting we have pointed out here that a
promising constructive tool
can be still found in the matrix continued fractions.
In the context of physics, one of the earliest physical
applications of MCFs
has been proposed
by Graffi and Grecchi \cite{GG}.
These authors decided to study the real and symmetric
pentadiagonal
(i.e., $M=2$)
Schr\"{o}dinger operator (\ref{finkit}) which
represented the most common (viz., quartic)
anharmonic-oscillator differential-operator
Hamiltonian
 $$
 H^{(QAO)}(g)= p^2 + q^2 + g\,q^4
 $$
in conventional harmonic-oscillator basis.
In a purely numerical setting
they noticed that in the infinite-dimensional-matrix
limit (with $N \to \infty$ in our present notation)
their two-by-two MCF expansion
defined by recurrences (\ref{cf}) proved convergent.

They also observed that a reasonably efficient
numerical search for the energy eigenvalues
can be based on the search of the zeros of
another Green's-function secular equation
of the form
  \be
  \det F^{-1}_1(E)=0\,.
  \label{effeq}
  \ee
From a purely pragmatic perspective, unfortunately,
the performance of the  Graffi's and Grecchi's
innovative MCF-based approach
did not exceed the efficiency of
several other, more standard numerical techniques.
One of the reasons was that
for their particular illustrative toy-model Hamiltonian $H^{(QAO)}(g)$
the rate of the $N \to \infty$ convergence of the MCF expansion of the
two-by-two
matrix $F^{-1}_1(E)$
in Eq.~(\ref{effeq}) happened to be
slow.

Naturally, this
did not imply a weakness of the MCF method itself.
For another family of certain
Fourier-symmetric anharmonic-oscillator
models $H^{(FSAO)}$ of paper \cite{SAO},
indeed, a much quicker rate of the practical numerical
convergence
has been achieved. It was also analytically proved there
for a virtually arbitrarily large
MCF dimension $M$.

Much less success has been achieved
in the less numerical quantum-physics-oriented
applications of the abstract mathematical formalism of MCFs.
In this sense, our present results can be perceived as encouraging.
For two reasons. Firstly, at least some of
the technical problems emerging in connection
with the use of complex (and, at the same time,
not necessarily Hermitian) tridiagonal Hamiltonians
were revealed to find solution in a transition to their
suitable non-tridiagonal
(i.e., in our present models, block-tridiagonal) partners.
Secondly, one should really be surprised by the survival
of feasibility of
technical analysis and, in particular, by the
efficiency of the specific fixed-point-based
proofs of the underlying continued-fraction convergence.

New perspectives are opening in several closely
related areas of research in
quantum physics (say, of unstable states) and in
complex analysis (and, in particular, in
mathematics of MCFs).
It is probably time to agree that
the current progress in the related functional analysis
(cf., once more, papers \cite{PS1,PS2})
seems to be paralleled by
the current progress in our understanding of
Schr\"{o}dinger equations with
local but complex potentials $V(x)$
(in this respect one just has to return, once again, to
their discretization via
Eq.~(\ref{vfinkit})).

\section*{Acknowledgement}

The author was supported by
Faculty of Science of
University of Hradec Kr\'{a}lov\'{e}

\end{document}